\documentclass[sigplan,nonacm]{acmart}
\usepackage{amsmath}

\usepackage{amssymb}
\usepackage{tikz}
\usepackage{colortbl}
\usepackage{listings}
\usepackage{subcaption}
\usepackage{diagbox}

\definecolor{resultzero}{HTML}{DDFFCC}
\definecolor{resultpartial}{HTML}{FFF2CC}
\newcommand{\resultzero}[1]{\cellcolor{resultzero}#1}
\newcommand{\resultpartial}[1]{\cellcolor{resultpartial}#1}
\newcommand{\codehdr}[1]{\textbf{#1:}\ }
\newcommand{\coderule}{\hrule height0.35pt}
\newcommand{\maxreconv}{\texttt{Maximally\allowbreak ReconvergesKHR}}
\lstdefinelanguage{GPU}{
  morekeywords={if,else,int},
  sensitive=true,
  morecomment=[l]{//},
  morecomment=[s]{/*}{*/},
  morestring=[b]",
}
\lstdefinestyle{gpufigure}{
  language=GPU,
  basicstyle=\ttfamily\scriptsize,
  keywordstyle=\color{blue},
  commentstyle=\color{gray},
  stringstyle=\color{red},
  numbers=left,
  numberstyle=\tiny\color{gray},
  stepnumber=1,
  numbersep=5pt,
  backgroundcolor=\color{white},
  tabsize=2,
  showspaces=false,
  showstringspaces=false,
  xleftmargin=2em,
}
\usetikzlibrary{arrows.meta,shapes.geometric,positioning}

\makeatletter
\def\@ACM@checkaffil{}
\makeatother

\newif\ifshowappendix\showappendixfalse
\newif\ifshowdiscussion\showdiscussiontrue

\begin{document}
\title{WarpDRF: The Unwritten Contract of GPU Warp Programming}
\author{Zheyuan Chen}
\orcid{0009-0001-7942-4193}
\affiliation{%
  \institution{University of California at Santa Cruz}
}
\email{zchen406@ucsc.edu}

\author{Simon Kagle}
\affiliation{%
  \institution{University of California at Santa Cruz}
}
\email{skagle@ucsc.edu}

\author{Tiago Cogumbreiro}
\affiliation{%
  \institution{University of Massachusetts Boston}
}
\email{Tiago.Cogumbreiro@umb.edu}

\author{Tyler Sorensen}
\affiliation{%
  \institution{University of California at Santa Cruz}
}
\affiliation{%
  \institution{Microsoft Research}
}
\email{tsorensen@microsoft.com}

\begin{abstract}
Warp primitives such as tensor core operations, shuffles, reductions,
and barriers are critical to high-performance GPU kernels, and every major GPU
language supports some set of them. The threads that participate in a
primitive, and therefore synchronize, are determined dynamically by how threads
diverge and reconverge, and by intra-warp scheduling such as independent thread
scheduling. In practice, many high-performance kernels use these primitives and
behave as expected, following an intuitive but unwritten data-race-freedom
contract that has never been stated precisely or empirically tested.

We present WarpDRF, the first abstract warp programming model to make this
contract precise, with participation rules parameterized by reconvergence
guarantees and per-primitive requirements so that instantiations match
different GPU languages. We prove (formalized in Rocq) that a kernel
satisfying WarpDRF executes every warp primitive with the participants the
reference semantics assigns and produces the same results, so a programmer can
reason in the reference semantics alone. We implement
the model in MLIR with a reference interpreter and fuzz 10K conformance tests
for each of three configurations across CUDA, HIP, HLSL, Metal, and SPIR-V on 16
device and backend pairs;
at least one WarpDRF configuration describes each, with CUDA honoring the
strictly weakest configuration. Finally, we extend Faial, a static data-race analyzer for
CUDA, into
the first DRF checker for an empirically validated warp model, and apply it to
llama.cpp, where most kernels that use warp primitives already satisfy the
contract, but three contain previously unknown data races, showing the need for
tools that check it.
\end{abstract}

\maketitle 
\section{Introduction}

To achieve high performance on GPUs, developers need to program the lowest
unit of parallelism available, called a \emph{warp}\footnote{ The term warp
originates from CUDA~\cite{cuda:programming-guide} and is also known as
wavefront in AMD's HIP~\cite{hip:programming-model}, and as subgroup in
Khronos specifications~\cite{spirv:spec}. }, which defines a group of threads
that the hardware schedules as a single unit, usually 32.
GPU programming models therefore offer low-level concurrency control, called
\emph{warp primitives}, that define data communication and synchronization
between the threads in a warp, such as parallel matrix operations,
register exchanges among threads, parallel reductions, and all-to-all
synchronization.
General Matrix-Matrix Multiplications (GEMMs) and Attention (a neural-network
operation)~\cite{attention:nips17} dominate the runtime of modern AI
workloads, and these operations leverage warp primitives to achieve high
performance~\cite{tensor-cores:icpe24}.
For instance, we found warp primitives in a fifth of the functions of a
popular linear-algebra library that powers AI workloads, that is, 6,061 of
27,383 disassembled function bodies of NVIDIA's cuBLAS use warp primitives.

A warp primitive is a collective operation, executed jointly by a subset of a
warp's threads (the \emph{participants}), and completes only once every
participant has reached the primitive.
The set of participants is not fixed by the primitive but by control flow,
since branching determines which threads reach the primitive.
The participant set in turn determines both the value the primitive returns
and which threads the primitive orders.
Because the returned value may decide subsequent control flow, and hence
subsequent participant sets, the functional behavior of a primitive and the
orderings the primitive establishes cannot be treated separately.

Crucially, languages differ in what a programmer may assume about the
participant set.
SPIR-V provides an optional execution mode that requires maximal reconvergence
for entry points~\cite{spirv:max-reconvergence}, so that the
participant set follows from the structure of control flow, whereas CUDA
offers no reconvergence guarantee and instead requires the programmer to name
the intended participants in a mask~\cite{cuda:programming-guide}.
The participant set moreover depends on scheduling, not only on control flow.
NVIDIA documents independent thread scheduling, under which threads of a warp
may interleave at fine granularity rather than execute in lockstep as
programmers long assumed~\cite{cuda:programming-guide,chen2026simtstep}, so
that threads on the same control-flow path need not reach a primitive
together.

Data-race freedom solved an analogous problem for memory ordering. Rather
than reason about every schedule and every relaxation the hardware permits,
programmers order conflicting accesses through high-level synchronization
constructs and, in return, may reason as if execution were sequentially
consistent~\cite{adve1990weak}. That discipline has been carried onto GPUs,
but the memory models and analyzers built there treat synchronization as
occurring among every thread at a fixed level of the hierarchy, whether a
block reaching a barrier or a scope reaching an atomic, and have no notion
of an execution-determined subset of a warp
synchronizing~\cite{betts2012gpuverify,hower2014hrf,faial:fmsd23,liew2024faialaa}.

A separate line defines participation and identifies where primitives
can execute~\cite{chen2026simtstep,leejones2026uniformity}, but offers no rule
for a programmer to follow. The two have developed independently, and the
intuition that keeps deployed kernels correct has never been stated as a
data-race-freedom contract, nor has the code that relies on it been checked
against one.

\begin{figure}[t]
\centering
\begin{subfigure}[b]{0.48\linewidth}
\begin{lstlisting}[style=gpufigure,numbers=left,numbersep=5pt,xleftmargin=1.4em]
// executed by
// one warp per block

// warp divergent branch
if (cond)
  // write
  shared[i] = ...;

shfl_sync(...);

// read
... = shared[j];
\end{lstlisting}
\caption{}
\label{fig:mmids-code}
\end{subfigure}
\hfill
\begin{subfigure}[b]{0.50\linewidth}
\centering
\begin{tikzpicture}[
  blk/.style={draw,rounded corners,inner sep=2pt,font=\scriptsize,
              minimum width=16mm,align=center},
  >=Stealth,->]
  \node[blk] (c) at (0,0)    {\texttt{if (cond)}\\[-2pt]{\tiny$P=\mathrm{all}$}};
  \node[blk] (w) at (0,-1.2) {write shared\\[-2pt]{\tiny$P=S$}};
  \node[blk] (s) at (0,-2.4) {shuffle\\[-2pt]{\tiny$P=\;?$}};
  \node[blk] (r) at (0,-3.6) {read shared\\[-2pt]{\tiny$P=\;?$}};
  \draw (c) -- (w);
  \draw (w) -- (s);
  \draw (s) -- (r);
  \draw (c.east) to[out=0,in=0,looseness=1.3] (s.east);
  \draw[<->,dashed] (w.west) to[out=180,in=180,looseness=1.2]
       node[midway,above,sloped,font=\tiny] {conflict?} (r.west);
\end{tikzpicture}
\caption{}
\label{fig:mmids-cfg}
\end{subfigure}
\caption{Warp participation and synchronization in a pattern abstracted from
llama.cpp's \texttt{mm\_ids\_helper}. In (b), $P$ denotes the threads that
execute each dynamic basic block. A divergent subset $S$ writes shared memory
before a warp shuffle, after which the warp reads the array. The shuffle does
not order shared memory, so the write and read conflict. In CUDA, the shuffle's
bitmask names the required participants; WarpDRF-Faial checks both the
participation requirement and the missing memory ordering.}

\Description{A kernel excerpt writes a shared array under a data-dependent
guard, executes a warp shuffle, and then reads the shared array. Beside it, a
control-flow graph shows four nodes labelled with their participating threads:
the guard with all threads, the write with a subset, and the shuffle and the
read each with an unknown set. An edge bypasses the write, and a dashed edge
marks the conflict between the write and the read.}
\label{fig:mmids}
\end{figure}

\subsection{WarpDRF: Contract, Conformance, Code}

We present WarpDRF, a data-race-freedom contract for warp programming. WarpDRF
is a simplified execution model that hides low-level detail such as lockstep
execution and independent thread scheduling, replacing it with two dependent
conditions on memory ordering and warp primitive participation. We do not
capture any one language in full, but target a pragmatic subset that maps
straightforwardly onto CUDA, HIP, HLSL, Metal, and SPIR-V. We build an
executable model in MLIR and test 16 device and backend pairs against it, with
CUDA honoring the weakest configuration. We then produce the first static
checker for
an empirically validated warp programming model and apply it to the
performance-critical CUDA kernels of llama.cpp, a widely used open-source LLM
inference engine, where most kernels already satisfy the contract and three contain
previously unknown data race.

\paragraph{Defining the contract}
We define WarpDRF over a deliberately small language: threads grouped into
warps, ordinary memory accesses, warp primitives, and control flow whose
conditions may differ across threads. Warp primitives operate over participants,
the threads that execute them together, and optionally order their memory
accesses.

A kernel satisfies WarpDRF when two conditions hold in its reference execution 
for every input. First, conflicting memory accesses must be ordered, and warp
primitives can optionally be set to provide that ordering. Second, every dynamic
execution of a warp primitive, which we call an instance, has an unambiguous set
of participants, ranging from all threads in the warp to any valid subset. The
conditions are dependent, since a primitive orders memory only for its
participants. Figure~\ref{fig:mmids} highlights a real llama.cpp kernel in
which insufficient synchronization causes a bug, and participation assumptions
are non-portable.

Each language instantiates its own configuration of WarpDRF by
specifying when threads reconverge and what else it requires of a primitive's
participants. From those choices we derive when a kernel satisfies the contract.
We prove a guarantee for the contract: on any conforming system, a kernel that
satisfies WarpDRF forms the same participant set at every primitive and
computes the same results as the reference execution, so a programmer can
reason in the reference semantics alone. We formalize the proof in Rocq over
SIMT-Step's synchronous subgroup operations (\textsf{SSO})~\cite{chen2026simtstep}, in
which threads advance independently and a primitive waits for every thread that
may still reach it.

\paragraph{Testing implementations against the contract}
Specifications say little about warp execution, so we test whether
implementations respect the contract. WarpDRF gives each test one expected
result, allowing its reference interpreter to serve as the oracle.

Across 10K tests for each of three configurations on 16 device and backend
pairs, every language conforms to at least one of them. Empirically, CUDA honors
the strictly weakest configuration,
and thus that configuration is portable across every system we tested. In it,
warp primitives are valid only when every thread in the warp executes
them, ruling out non-uniform participation. Some languages additionally define
partial-participation behavior, but those guarantees are specific to the
language and execution mode.

\paragraph{Checking kernels against the contract}
We build a checker for the CUDA configuration by extending Faial, a static
data-race
analyzer that previously had no support for warp
primitives~\cite{faial:fmsd23}. We run it over llama.cpp and check the 95 kernels
that use warp primitives. After checking every reported alarm
against the kernel source and the current callers and launch configurations,
we find that 77 of the 95 analyzed kernels (81.1\%) have no unordered
conflicting memory accesses and execute every warp primitive with the full
warp.
The remaining 18 kernels map to four broader source patterns containing confirmed or input-conditional alarms (\S\ref{sec:llama-eval}).
\paragraph{Contributions}
In summary:

\begin{enumerate}
\item \textbf{WarpDRF} (\S\ref{sec:warpdrf-contract}): a data-race-freedom
contract in which every instance of a warp primitive has a single unambiguous
set of participants and every memory conflict is ordered, with configurations
for languages that reconverge differently, and a guarantee, formalized in Rocq,
that satisfying kernels match the reference semantics.

\item \textbf{An executable model} (\S\ref{sec:realizing-warpdrf}): an MLIR
dialect and interpreter to fuzz 10K tests for each of three
configurations over five languages and 16
device and backend pairs (\S\ref{sec:cross-platform}), where CUDA honors the
weakest configuration.

\item \textbf{A warp-aware race checker}
(\S\ref{sec:extended-faial-protocols}): the first built on an empirically
validated model, applied to 95 llama.cpp kernels
(\S\ref{sec:llama-eval}). After source and caller validation, 77 (81.1\%)
satisfy both conditions; the other 18 (18.9\%) are buggy, with confirmed or
input-conditional alarms across four source patterns.
\end{enumerate}

\section{Background}

\subsection{Execution Hierarchy and Synchronization Scopes}
\label{sec:background-scopes}

A GPU kernel executes a grid of logical threads organized into hierarchical
groups. CUDA and HIP call the larger group a \emph{thread block},
Metal calls it a \emph{threadgroup}, and SPIR-V calls it a \emph{workgroup}.
Implementations partition each group into tightly coupled execution
groups: CUDA uses \emph{warps}, AMD hardware uses \emph{wavefronts}, Metal uses
\emph{SIMD-groups}, and SPIR-V uses \emph{subgroups}
~\cite{cuda:programming-guide,hip:programming-model,metal:threadgroups,spirv:spec}.
We use \emph{block} and \emph{warp} for these concepts. We use \emph{warp
primitive} as an umbrella term for two categories. A \emph{warp collective},
such as a shuffle, ballot,
reduction, or cooperative matrix instruction, exchanges or combines values
across threads. A \emph{warp barrier} waits for participating threads and may
provide warp-scoped memory ordering.

The execution hierarchy determines the scope of communication and
synchronization. Registers are private to a thread, shared or workgroup memory is
addressable within a block, and global memory is addressable device-wide. A
block barrier such as \texttt{\_\_syncthreads()} coordinates the
threads in one block and orders their relevant memory accesses. A
memory-ordering warp barrier such as \texttt{\_\_syncwarp()} has narrower scope:
it coordinates and orders accesses only among threads of one warp. It cannot
establish ordering between different warps, even when they belong to the same
block~\cite{cuda:programming-guide,spirv:spec}.

Collective participation and memory ordering are separate effects. Shuffles,
ballots, and reductions in CUDA and HIP, their Metal SIMD-group counterparts,
and SPIR-V group-nonuniform operations exchange or aggregate register values but
do not by themselves order accesses to shared or workgroup memory. In
particular, the \texttt{\_sync} suffix on CUDA and HIP collectives identifies
participating threads; it does not turn the collective into a memory barrier
~\cite{cuda:programming-guide,hip:language-extensions,
metal:shading-language,spirv:spec}. Memory ordering instead requires an explicit
barrier carrying the appropriate scope and memory semantics. SPIR-V makes this
separation syntactic. Group-nonuniform operations carry only an execution scope,
while control barriers also carry a memory scope and memory semantics. Similarly,
Metal uses flags to select the memory classes ordered by a SIMD-group barrier.
WarpDRF therefore treats these warp collectives as participation events rather
than memory boundaries. Only barriers with documented memory semantics provide
scoped memory ordering.

\subsection{Divergence and Collective Participation}
\label{sec:background-participation}

Threads diverge when a thread-dependent branch or loop sends them along different
dynamic control-flow paths. An implementation may serialize those paths,
predicate instructions, allow threads to progress independently, or combine these
techniques. Reconvergence describes when the threads resume joint execution, but
GPU languages provide different guarantees about where this occurs. SPIR-V
identifies structured merge points, and a shader that declares the
\maxreconv{} execution mode requires threads to rejoin at
them as early as the program structure allows, whereas general CUDA does not
promise an analogous rejoining at arbitrary source-level
joins~\cite{cuda:programming-guide,spirv:spec}. SIMT-Step formalizes these
alternatives using dynamic blocks~\cite{chen2026simtstep}.

Table~\ref{tab:language-participation} summarizes how each API determines the
participants of a warp primitive under divergent control flow.

\begin{table*}[t]
  \centering
  \scriptsize
  \setlength{\tabcolsep}{3pt}
  \renewcommand{\arraystretch}{1.08}
  \caption{Selected specification-level participation rules for warp-scoped
  primitives
  ~\cite{cuda:programming-guide,hip:language-extensions,
  metal:shading-language,spirv:spec,spirv:max-reconvergence}. An explicit mask
  selects participants for one primitive; it does not provide a general
  reconvergence guarantee.}
  \label{tab:language-participation}
  \begin{tabular}{@{}>{\raggedright\arraybackslash}p{0.09\textwidth}
                      >{\raggedright\arraybackslash}p{0.18\textwidth}
                      >{\raggedright\arraybackslash}p{0.27\textwidth}
                      >{\raggedright\arraybackslash}p{0.38\textwidth}@{}}
    \toprule
    Model & Primitive class & Participants & Divergent-control guarantee \\
    \midrule
    CUDA & Masked collectives
      & Calling threads named by one consistent mask
      & Disjoint masks may execute divergently; every non-exited named thread must reach the same intrinsic with the same mask \\
    \midrule
    HIP & Cross-thread \texttt{\_sync} collectives
      & Active calling threads named by one consistent 64-bit mask
      & Every active named thread must execute the same collective with the same mask \\
    \midrule
    Metal & SIMD-group collectives
      & Active threads
      & Results range over the active threads at the dynamic call \\
    \midrule
    SPIR-V & Group-nonuniform operations
      & Invocations in the same scope-restricted dynamic instance
      & Structured merges can restore uniform flow; \texttt{Maximally\-ReconvergesKHR} further constrains dynamic grouping \\
    \midrule
    CUDA, Metal, SPIR-V & Cooperative matrix
      & Complete warp, SIMD-group, or matrix scope
      & Uniform execution is required \\
    \bottomrule
  \end{tabular}
\end{table*}

All four APIs support complete-group execution but differ on partial
participation. Section~\ref{sec:instantiating} represents these differences as
target configurations.

One source-level warp primitive may execute multiple times, for example when it
appears in a loop or a repeated function call. We call each execution a
\emph{primitive instance}. Each instance has its own participant set, and a
collective's result depends on both those participants and the values they
supply. Threads reaching the same source location through different loop
iterations or calls therefore need not belong to the same instance.
Section~\ref{sec:abstract-executions} formalizes these instances and their
participant sets.

Cooperative matrix operations make the full-group requirement explicit: CUDA WMMA
requires all threads in a warp to execute each operation, while SPIR-V requires all
invocations in the cooperative-matrix scope to be active or none
~\cite{cuda:programming-guide,spirv:spec}. Neither specification assigns portable semantics to
execution by a proper subset, which motivates WarpDRF's fully convergent contract.

\subsection{Why Memory DRF Is Insufficient}
Two ordinary memory accesses conflict when they touch the same location and at
least one is a write. A conflicting pair is a data race unless one access
happens before the other. To construct this relation, a DRF model begins with
the synchronization operations in an execution and the ordering that each one
provides. For example, a block barrier is a block-scoped synchronization
operation by definition, so it can order accesses by threads in that block.
GPU DRF models account for the hierarchy by restricting such ordering to scopes
such as a workgroup or device~\cite{hower2014hrf,faial:fmsd23}.

Warp primitives complicate this starting point. Unlike a block barrier, a warp
primitive may execute under divergent control flow, and any coordination or
memory ordering applies only to threads that execute it together.
A memory-race analysis must therefore determine the group
before using a warp primitive to order accesses.

Figure~\ref{fig:two-thread-barriers} illustrates how divergence and
reconvergence determine that group, and which memory accesses the resulting
barrier orders. It uses a warp of exactly two threads, 0 and 1. The array
\texttt{g} is shared, while \texttt{local} and \texttt{local2} are local
variables of each thread. Here \texttt{warp\_sync} orders memory among its
participants. Assume barriers permit partial participation and the threads
rejoin after the conditional. A therefore contains only thread 0, whereas B
contains both threads. Reconvergence determines the group at B; it does not
itself provide B's memory ordering.

\begin{figure}[t]
  \centering
  \begin{minipage}{0.96\linewidth}
    \centering
    \codehdr{Initial}\texttt{g[0] = g[1] = 0;}
    \coderule
    \begin{lstlisting}[style=gpufigure]
g[tid] = tid + 1;
if (tid == 0) {
    warp_sync();          // A: thread 0.
    local = g[1];
}
warp_sync();              // B: threads 0 and 1.
local2 = g[1 - tid];
    \end{lstlisting}
    \coderule
  \end{minipage}
  \caption{Barrier placement in a two-thread warp, assuming barriers permit
  partial participation and the threads rejoin after the conditional. A
  synchronizes only thread 0. B orders the initial writes before the final
  reads, but neither barrier orders thread 1's write before the branch read.}
  \Description{A program for threads 0 and 1, with both entries of shared
  array g initially zero. Each thread writes its own entry.
  Thread 0 synchronizes alone and reads thread 1's entry,
  then both threads synchronize and read each other's entries.}
  \label{fig:two-thread-barriers}
\end{figure}

Thread 1's write \texttt{g[1] = 2} and thread 0's branch read
\texttt{local = g[1]} access the same location. A supplies no memory ordering
between these threads, because thread 1 does not participate. Both accesses
precede B, so B does not order them either: the write and read race. In
contrast, B orders each initial write before the other thread's final read
into \texttt{local2}. Removing A leaves the memory ordering unchanged;
removing B also leaves those final reads unordered with the other thread's
writes.

Removing A and moving B immediately after the initial writes, before the
conditional, orders every conflicting write--read pair. The branch can remain
divergent, and no barrier is needed afterward because only reads remain.
The moved B also runs in warp-uniform control, so its group no longer depends
on reconvergence. The original A runs with only part of the warp, which is
invalid if the primitive requires full-warp participation.

The participant set also matters for primitives that do not order memory: a
collective reached by only part of a warp can be memory-race-free while
producing a register value that depends on how the threads are grouped.

WarpDRF establishes a portable fully convergent domain by requiring every
executed warp primitive to include the complete warp. This rule captures a
common kernel idiom rather than requiring the whole kernel to be uniform: threads
may compute different operands and take different paths between primitives, but
invoke each warp primitive in warp-uniform control.
The moved B in Figure~\ref{fig:two-thread-barriers} follows this
rule; A does not.
This fixes the participants to which a warp primitive applies, without relying
on API-specific rules for reconverging divergent threads. WarpDRF then checks
that conflicting memory accesses are ordered at an appropriate scope.
Section~\ref{sec:warpdrf-contract} formalizes this full-warp contract and
separately describes how to reason about programs that intentionally rely on
partial participation.

\section{The WarpDRF Contract}
\label{sec:warpdrf-contract}

WarpDRF retains the classic data-race-freedom argument: programmers use
high-level synchronization to order conflicting memory accesses rather than
reason about low-level execution schedules~\cite{adve1990weak}. What this
work adds is a definition of that synchronization. Classic models take their
synchronization operations as given: a lock or a block barrier is a named
operation whose ordering rules are stated in advance. For a warp primitive,
which calls make up a single operation is in question, because any coordination or memory ordering applies only to 
threads that execute it together.

WarpDRF therefore fixes that grouping before using it. Its reference semantics
identifies the threads that execute each primitive together, which makes
participation predictable under divergent control flow. Each language then
instantiates the contract: from its specification, such as SPIR-V's maximal
reconvergence, we derive which participant groups that language defines. Once
those groups are fixed, ordering conflicting accesses is the classic argument.
Language specifications are often informal, so these derivations are
best-effort. In Section~\ref{sec:realizing-warpdrf} we design a conformance
testing methodology and show that some variant of them holds across every
language and device we tested.

\subsection{The WarpDRF Language}
\label{sec:core-language}

Figure~\ref{fig:core-language} defines the small language over which WarpDRF is
stated. It retains the three features relevant to the contract: memory
accesses, warp primitives, and structured control flow whose conditions may
evaluate differently across threads. Scalar values are integers, and any other
pure thread-local computation is an uninterpreted function \(f\): calls to the
same \(f\) with equal arguments produce equal results, and the semantics
assumes no other property. This represents operations such as floating-point
arithmetic without defining their numerical behavior, keeping the formalization
focused on participation, memory ordering, and control flow.

A call to \(p\) denotes a warp primitive and optionally returns a value. Given
a set of participants, which the reference semantics defines, its result is an
uninterpreted function over the values those participants provide.

Languages and devices differ, so our formalism is parameterized by a
configuration \(C\) that aims to capture those differences. A configuration
supplies:
\begin{itemize}
\item the warp width, which is 32 on many systems, famously NVIDIA, but can be
  any number; some devices are as large as 128;
\item \(\mathsf{ord}(p)\), whether a primitive \(p\) orders memory among its
  participants;
\item \(\mathsf{req}(p)\), which participant sets are legal for \(p\); and
\item how participation is defined under divergence.
\end{itemize}

In CUDA, \texttt{\_\_shfl\_sync} only exchanges values and has
\(\mathsf{ord}(p)=\mathsf{nosync}\), whereas \texttt{\_\_syncwarp} has
\(\mathsf{ord}(p)=\mathsf{sync}\). Tensor core operations have
\(\mathsf{req}(p)=\mathsf{full}\), requiring every thread in the warp, while
shuffles and votes may be called with a partial group. The last item is the
least settled: SPIR-V's optional \maxreconv{} mode requires
reconvergence according to its structured rules, whereas CUDA provides no
analogous general guarantee, which we investigate empirically in
Section~\ref{sec:cross-platform}.

Following SIMT-Step,
the language assumes structured control flow
with explicit merge points~\cite{chen2026simtstep}; unstructured control flow is
outside this formalization. None of the 64 CUDA source files contributing the
479 kernels in our study uses \texttt{goto}, so this restriction
excludes no source-level control flow encountered in that evaluation. Atomics,
asynchronous copies, and cross-block synchronization are also outside the
language. These features occur in only six of the 479 kernels, none of which
uses a warp primitive, so they exclude no kernel from the 95-kernel evaluation
set. Prior work shows how relaxed GPU atomics can be incorporated into a
data-race-freedom model~\cite{sinclair2017rats}; we leave that extension to
future work.

\begin{figure}[t]
\footnotesize
\[
\begin{array}{r@{\;}c@{\;}l@{\quad}l}
\multicolumn{4}{@{}l}{
  n \in \mathit{Int} \;\; x \in \mathit{Var} \;\; a \in \mathit{Arr} \;\;
  f \in \mathit{Fun} \;\; p \in \mathit{Prim}} \\
\multicolumn{3}{@{}l@{\quad}}{
  \mathsf{ord} : \mathit{Prim} \to \{\mathsf{sync},\mathsf{nosync}\}}
  & \text{orders memory?} \\
\multicolumn{3}{@{}l@{\quad}}{
  \mathsf{req} : \mathit{Prim} \to \{\mathsf{full},\mathsf{any}\}}
  & \text{participant constraint} \\[7pt]
e &::=& n \mid x \mid \mathsf{tid} \mid e \oplus e \mid f(e,\ldots,e)
  \\[7pt]
s &::=& x := e & \text{local assignment} \\
  &\mid& x := a[e] \mid a[e] := e & \text{memory access} \\
  &\mid& [\,x := ]\; p(e,\ldots,e) & \text{warp primitive} \\
  &\mid& \mathsf{if}\;e\;\mathsf{then}\;s\;\mathsf{else}\;s
  & \text{conditional} \\
  &\mid& \mathsf{switch}\;e\;\{\,n \Rightarrow s,\ \ldots\,\}
  & \text{multiway branch} \\
  &\mid& \mathsf{while}\;e\;\mathsf{do}\;s & \text{loop} \\
  &\mid& s\,;\,s \mid \mathsf{skip} \mid \mathsf{break} \mid
         \mathsf{return} & \text{sequencing, exit}
\end{array}
\]
\caption{The WarpDRF language. Here \(\mathsf{tid}\) is the linearized thread
identifier, \(\oplus\) ranges over integer arithmetic and comparisons, and a
primitive's result binding is optional.}
\Description{A BNF grammar for the WarpDRF language. Metavariables range over
integers, local variables, arrays, uninterpreted function symbols, and warp
primitives. A function ord maps each warp primitive to sync or nosync, and a
function req maps each warp primitive to full or any.
Expressions are integer literals, variables, the thread identifier, binary
operations, and applications of uninterpreted function symbols. Statements are
local assignment, array read, array write, warp primitive call with an optional
result binding, conditional, switch, loop, sequencing, skip, break, and
return.}
\label{fig:core-language}
\end{figure}

\paragraph{Correspondence to real languages}
Our language is intentionally small and general. Given the subtle differences
between real GPU languages, we do not present it as a complete formalization of
any one of them, but it covers the constructs used by many impactful kernels
(Section~\ref{sec:llama-eval}). The mapping is most direct for SPIR-V, Metal,
and HLSL: array accesses become loads and stores, a warp corresponds to a
subgroup or SIMD-group, and primitive calls become the corresponding subgroup
operations. CUDA and HIP further allow a call to name its participants with a
bitmask. We do not model bitmasks, since no other language supports them; in
practice the value passed is frequently \texttt{\_\_activemask()}, which reports
whichever threads control flow has already brought to the call rather than
requiring a particular set. Section~\ref{sec:instantiating} examines what CUDA
guarantees about that set, and Section~\ref{sec:test-pipeline} implements these
mappings.

\paragraph{Reference semantics}
We adopt the bulk-synchronous, block-level semantics used by earlier GPU
analysis tools~\cite{betts2012gpuverify,li2012gklee,collingbourne2013interleaving}:
one thread runs until it
reaches a barrier, then another runs until it reaches a barrier, and the
barrier is released once every thread has arrived. There is no interleaving,
only a sequentialized execution of barrier intervals. WarpDRF derives
participation from control flow rather than from a declaration and extends this
semantics to warps by replacing block barriers with the participants of a warp
primitive.

The difficulty is that, unlike a block barrier, a warp primitive admits partial
participation. Prior work describes correct use of block barriers with notions
such as textual alignment~\cite{kamil2006concurrency}, which asks that every
thread reach the same barrier in the same loop iteration. To extend that idea
to warps we borrow
\emph{dynamic basic blocks}~\cite{chen2026simtstep}, which lift SPIR-V's notion
of a dynamic instruction instance~\cite{spirv:dynamic-instance} to a basic
block: each execution of a basic block
creates a new dynamic block, and a dynamic block is executed by a group of
threads rather than by one. The entry point of the program is a dynamic block
executed by all threads, and a branch creates dynamic blocks executed by
the threads that take the same path.

Dynamic blocks give us the machinery to define participation. If a warp
primitive occurs in a dynamic block executed by the thread set \(T\), and \(C\)
permits that primitive to occur there, then its participants are \(T\).
Distinct loop iterations and repeated calls produce distinct dynamic blocks, so
threads that reach the same source location by different paths are not grouped
automatically. Section~\ref{sec:cuda-sensitivity} tests this behavior.

This extends the bulk-synchronous semantics of earlier GPU
analyses~\cite{betts2012gpuverify,li2012gklee}: SIMT-Step's dynamic
blocks~\cite{chen2026simtstep} determine each warp primitive's group instead of
assuming that every synchronization operation has a fixed group.

The reference semantics therefore picks threads in turn and runs each until it
reaches a warp primitive, updating the set of dynamic blocks as it goes. A
primitive is released once all of its participants have arrived, which the
semantics cannot determine until the paths of the remaining threads are
resolved. Threads are picked in a fixed order, so a kernel has exactly one
reference execution for a given input.

\subsection{The Contract}
\label{sec:abstract-executions}

WarpDRF states its conditions over the reference semantics rather than over the
executions a target system may produce. Classic DRF is stated the same way: a
program is data-race-free if no sequentially consistent execution has a race,
and a guarantee then transfers that reasoning to real implementations. Both
conditions are properties of the kernel. First, conflicting memory accesses
must be ordered. Second, every warp primitive must run with a group that the
configuration \(C\) guarantees. The second condition is new. Without it a
kernel could rely on a grouping that \(C\) leaves open, and a target system
free to group threads differently would change either the memory ordering or
the result of a primitive.

Write \(E\) for the reference execution of a kernel on a given input, and let
\(M_E\) be the memory events in \(E\). Every array read or write produces an
event recording its thread, location, and access mode, either read or write.
Every executed warp primitive has the participant set determined by its dynamic
block, as defined in Section~\ref{sec:core-language}.

The happens-before relation \(\mathsf{hb}_E\) is transitive, and contains each
thread's program order and the ordering added by warp primitives with
\(\mathsf{ord}(p)=\mathsf{sync}\), which applies only among that primitive's
participants.

\paragraph{Condition 1: memory DRF}
Two events in \(M_E\) conflict when they come from different threads, access the
same location, and at least one is a write. An execution is
memory-data-race-free when every conflicting pair is ordered:
\[
\begin{aligned}
\mathsf{MemDRF}(E)\triangleq \forall e,f\in M_E.\quad&
  \mathsf{conflict}(e,f)\Rightarrow{} \\
& e\mathrel{\mathsf{hb}_E}f
  \lor f\mathrel{\mathsf{hb}_E}e.
\end{aligned}
\]
This is the standard DRF condition. WarpDRF changes neither the definition of a
conflict nor the requirement that conflicts be ordered. Its effect on this
condition is that ordering from a warp primitive applies exactly to
that primitive's participants.

\paragraph{Condition 2: unambiguous participation}
Participation is unambiguous when every warp primitive in the reference
execution has a participant group that \(C\) guarantees, and that group
satisfies the primitive's constraint
\(\mathsf{req}(p)\). Which groups those are depends on the configuration;
Section~\ref{sec:instantiating} defines a portable complete-warp case and
language-specific cases that admit partial participation.

\paragraph{Combining the conditions}
A kernel satisfies WarpDRF for configuration \(C\) when its reference execution
is memory-data-race-free and its participation is unambiguous. The two
conditions are dependent: \(\mathsf{hb}_E\) is built from warp primitives, so
the memory condition cannot be evaluated until participation is settled. Even
for
a primitive that does not order memory, participation still matters, because
changing the group can change the primitive's result.
Figure~\ref{fig:mmids} shows the two conditions interacting in a real kernel.

\subsection{The Guarantee}
\label{sec:guarantee}

The guarantee is the classic DRF one~\cite{adve1990weak}: a kernel that
satisfies WarpDRF produces the same reads and the same final memory on a target
system
as it does in the reference semantics, so a programmer can reason in the
reference semantics alone. Earlier GPU semantics establish analogous agreement
between lockstep and interleaving executions for race-free
kernels~\cite{collingbourne2013interleaving,kojima2017hoare}.

The novel component of WarpDRF is that participation must be settled before the
memory argument can run. A target system is \emph{conforming} when it behaves as
its configuration \(C\) describes, does not let a primitive instance take effect
until every thread that may still reach it has arrived or taken a path that
excludes it, provides the ordinary DRF
guarantee with warp primitives supplying the synchronization edges among their
participants, and preserves the specified behavior of operations under
compilation. The last two are standard.

\begin{theorem}[WarpDRF guarantee]
\label{thm:warpdrf-guarantee}
If a kernel satisfies WarpDRF for configuration \(C\), then on a conforming
target system it executes each warp primitive with the participants that the
reference semantics assigns, and produces the same reads and the same final
memory as the reference execution.
\end{theorem}

The theorem does not require every participant group to be a complete warp. If
\(C\) enables maximal reconvergence in SPIR-V, the theorem covers the partial
groups guaranteed by that mode, provided the reference semantics determines
those groups.

\paragraph{Instantiating conformance}
Theorem~\ref{thm:warpdrf-guarantee} is stated against conformance, a property
we require of a target, not one we establish. Before outlining a proof we
therefore fix a base semantics, instantiating conformance with a concrete warp
execution model. SIMT-Step provides a family of such models for a single
warp, totally ordered by strength, where a stronger model restricts when a
primitive may execute and therefore admits fewer
executions~\cite{chen2026simtstep}. We examine two members of this family. The
first is too weak to support WarpDRF, and the next strongest is sufficient.
Both use the same memory model, which we
deliberately make weak. Following GPUVerify, a write stays private to its
thread until that thread participates in a memory-ordering primitive, and an
execution aborts as soon as it races~\cite{betts2012gpuverify}. We call this
\emph{delayed visibility}.

\paragraph{Speculative execution is too weak}
The weaker of the two is speculative subgroup operations (\textsf{Spec}), which
executes a primitive with the threads currently known to participate,
speculating that no additional threads will join.
Figure~\ref{fig:spec-counterexample} shows that this freedom is enough to break
the guarantee. Threads 0 and 1 form the warp and \(g[0]\) starts at
0. Each thread reads \(g[0]\) and executes a memory-ordering
\texttt{warp\_sync} if it read 0, after which thread 0 writes 1. Write
\(R_t(v)\) and \(W_t(v)\) for a read and a write of \(g[0]\) by thread \(t\),
and \(S(P)\) for the primitive with participants \(P\).

\begin{figure}[t]
  \centering
  \begin{minipage}{0.9\linewidth}
    \centering
    \codehdr{Initial}\texttt{g[0] = 0;}
    \coderule
    \begin{lstlisting}[style=gpufigure]
int r = g[0];
if (r == 0) {
    warp_sync();
    if (tid == 0)
        g[0] = 1;
}
    \end{lstlisting}
    \coderule
  \end{minipage}
  \caption{Speculative execution can violate the WarpDRF guarantee. This
  adaptation of SIMT-Step's Figure~10~\cite{chen2026simtstep} uses ordinary
  memory accesses and one writer. Threads 0 and 1 form the complete warp;
  \texttt{warp\_sync} orders memory among its participants.}
  \Description{Both threads read g[0] and, if it is zero, execute warp\_sync.
  Thread 0 then writes one. Under Spec, thread 0 can synchronize alone and
  write before thread 1 reads, so thread 1's read races with the write.}
  \label{fig:spec-counterexample}
\end{figure}

The reference execution is \(R_0(0), R_1(0), S(\{0,1\}), W_0(1)\). Thread 0
waits until thread 1 has decided whether to join, so both read 0 and both
participate. The only cross-thread conflict is thread 1's read against thread
0's write, and the primitive orders it. The kernel satisfies both conditions in
either configuration. \textsf{Spec} also permits \(R_0(0), S(\{0\}), W_0(1), R_1(0)\).
Here thread 0 executes the primitive alone while thread 1 has not yet decided,
then writes 1. That write is still private, so thread 1 reads 0, and the
primitive did not include thread 1, so nothing orders the write before the
read. The two race and the execution aborts. Thread 1 never reaches the
primitive, so \textsf{Spec}'s assumption that it would not join is never contradicted,
and the execution is legal. A kernel that satisfies both conditions can
have a \textsf{Spec} execution that races and forms a group the reference never
assigns. Therefore \textsf{Spec} does not provide a sound base semantics for
WarpDRF.

\paragraph{Synchronous subgroup operations}
The next strongest model is synchronous subgroup operations (\textsf{SSO}). An instance
executes only when every thread that may still reach it is already waiting
there. Deciding which threads those are is the delicate part of the semantics:
the test inspects both arms of an unresolved branch and the relevant loop
iterations, and it does not predict the result of a future read. The rule
constrains only the moment at which a primitive executes. Between primitives,
threads advance independently and in any order, so \textsf{SSO} is far weaker than
lockstep and admits the fine-grained interleaving that independent thread
scheduling exposes. \textsf{SSO} is nonetheless strong enough to support WarpDRF, and
it remains strong enough under delayed visibility.

\begin{theorem}[\textsf{SSO} supports WarpDRF]
\label{thm:sso-supports}
Fix an input to a single-warp kernel whose reference execution terminates and
satisfies both WarpDRF conditions, for either the fully convergent or the
partially convergent configuration. Then no \textsf{SSO} execution with delayed
visibility races, and every terminating one assigns the same participants to
each instance, returns the same reads to each thread, and produces the same
final memory as the reference.
\end{theorem}

We mechanize the theorem in Rocq, covering structured loops, early exits,
primitive arguments and results, and both choices of
\(\mathsf{ord}(p)\).

\begin{proof}[Proof outline]
\renewcommand{\qedsymbol}{}
Fix a finished reference execution for the input, which condition 1 makes
memory-data-race-free. The argument matches any \textsf{SSO} execution with
delayed visibility against it, in two parts.

The first part rearranges reference steps, which is needed because \textsf{SSO}
lets threads advance in any order while the reference fixes one. An action
enabled in a state remains enabled after a different action runs, because a
thread excluded from an instance cannot become one of its participants later.
Two adjacent independent actions can then be exchanged without changing any
recorded event, where reference memory DRF rules out the case that would
otherwise block the exchange, a conflicting pair of accesses with nothing
between them to order it. Together these move any action enabled at the start
to the front, so any finite target prefix can be matched by a permutation of
the reference execution.

The second part relates delayed visibility to the reference memory, in which a
write takes effect at once. Let logical memory be the result of applying every
pending write to shared memory in issue order. We induct over target steps,
maintaining that the two executions agree on thread code and local values, that
logical memory agrees with reference memory at every address, and that all
pending writes to one address belong to a single thread. The key fact is that a
pending write cannot conflict with another thread's access before its author
participates in a memory-ordering primitive, and this follows from reference
DRF alone. Read agreement therefore follows from the reference being race-free,
and we never assume that the target is race-free. Primitive readiness,
participants, and arguments depend only on thread code and local values, so
they agree as well, and the joint step returns the same result to each
participant.

The two parts combine as follows. No target execution races, because its events
match those of a prefix of a memory-DRF reference execution, so none aborts.
The waiting rule gives each instance the participants the reference assigns,
which condition 2 places among the groups \(C\) guarantees. Publishing the
remaining pending writes at termination makes the final memories agree. Every
terminating \textsf{SSO} execution therefore agrees with the reference on the
participants of each instance, the reads of each thread, and the final memory,
which establishes the theorem.
\end{proof}

\subsection{Instantiating Configurations}
\label{sec:instantiating}

The contract does not change from one target system to another; the
configuration does. Two choices dominate an instantiation. One is per
primitive: which operations order memory among their participants, and which
require the complete warp. The other is reconvergence: which join points bring
threads back together, and therefore which divergent regions may contain a
primitive that runs with a partial group.
Table~\ref{tab:language-participation} summarizes these choices for the APIs
studied here, and Section~\ref{sec:background-scopes} describes their
memory-ordering rules.

\paragraph{SPIR-V}
Our partially convergent SPIR-V configuration requires an entry point to
declare the optional \maxreconv{} execution mode. For such
an entry point, SPIR-V requires invocations to reconverge according to the
extension's structured rules~\cite{spirv:max-reconvergence}. These rules fix
selection and loop merges as well as loop continue targets.
\texttt{OpSwitch} retains additional freedom: invocations with the same
selector remain together, while invocations with different selectors that
reach the same case may or may not be grouped together. Figure~\ref{fig:core-language}
therefore represents \texttt{switch} directly rather than translating it into
nested conditionals.

Among primitives, SPIR-V's cooperative matrix operations are the ones with
\(\mathsf{req}(p)=\mathsf{full}\): they require every invocation in the matrix
scope, which for the physical warps WarpDRF considers is the complete warp.

\paragraph{CUDA and HIP}
As Section~\ref{sec:core-language} notes, we do not model bitmasks. Our tests
pass \texttt{\_\_activemask()}, which queries the participants dynamically and
so corresponds more closely to the SPIR-V mapping. In pilot experiments we
could not find any reconvergence rule that reliably gave
unambiguous participation for primitives with partial groups. We regard this as
a serious problem with \texttt{\_\_activemask()}. Code that uses it to group a
partial warp may work today, but we did not observe any rule that
says it will keep working.

Despite these surprisingly unconstrained behaviors around
\texttt{\_\_activemask()}, one rule did hold empirically, even if we could not
find it stated in the CUDA specification. At a join point where every thread of
the warp reconverges, all reliably participate in a subsequent
primitive. This configuration is weaker than the SPIR-V one and is
therefore portable: requiring every primitive to run with the complete warp
leaves nothing for a reconvergence policy to disagree about.

SPIR-V and CUDA therefore give two configurations worth testing broadly, which
Section~\ref{sec:realizing-warpdrf} does. Section~\ref{sec:faial} then
implements the portable one, which that testing shows to be empirically valid
for CUDA, in the Faial race checker.

\section{Testing Implementations}
\label{sec:realizing-warpdrf}

WarpDRF's guarantee requires a conforming target system. We ask if a system
conforms, focusing on participation, the novel component of WarpDRF and the
property languages document the least. The contract specifies the reference
semantics, which can serve as an oracle. The oracle is efficient to compute
because WarpDRF has a single output for a given input; there is no need to
enumerate interleavings. We build this executable model in MLIR and test 16
device and backend pairs against it. We check the systems of
Table~\ref{tab:device-fuzzing} against the three configurations of
Section~\ref{sec:instantiating}. Checking every configuration establishes
whether a system conforms to any instantiation of WarpDRF, and the cases where
they disagree provide interesting insights into how systems implement warp
executions.

The conformance campaign generates kernels with non-obvious participation,
translates each to the target languages, and compares target results with the
interpreter under different participation rules. This separates expected
behavior from the compiler and hardware paths being tested.

\subsection{Conformance Tests and Oracle}
\label{sec:test-pipeline}

To test conformance, we implement the language of
Section~\ref{sec:core-language} as the WarpDRF MLIR dialect. Its operations are
the memory accesses, warp primitives, and structured control flow of
Figure~\ref{fig:core-language}. This dialect could be used as the foundation
for more analysis, but for now we use it for conformance testing.

To provide an oracle, we implement an interpreter for the WarpDRF dialect.
Computing the expected result follows the schedule of the model itself: execute
each thread in turn, switching at a warp primitive and releasing it once all
its participants arrive. The interpreter creates a dynamic execution graph
that tracks dynamic blocks, similar to SIMT-Step~\cite{chen2026simtstep}, and
updates each block's participation list as threads branch.

To provide concrete tests, we implement raisers to the different GPU languages
of Table~\ref{tab:device-fuzzing}. Control flow and expressions map naturally to
each language. The
GLSL path emits a shader embedded in an Amber script, which Amber compiles to
SPIR-V and executes through Vulkan. We note that HLSL does not have a mapping
to synchronizing warp
primitives; however, our tests only deal with participation, so it does not
impact this investigation.

\paragraph{A WarpDRF conformance test}
A test is a WarpDRF dialect module and a launch configuration giving the number
of threads and the warp width. It initializes arrays that contain values for the
test. For each target, we launch exactly one complete warp: if the target warp
width is $W$, the launch contains $W$ threads. As in SIMT-Step, we obtain
$W$ from the target rather than fixing one width for every
device~\cite{chen2026simtstep}.

Participation is observed through \texttt{wave\_count\_bits}, a warp primitive
of the WarpDRF language that returns the number of threads participating in it;
the artifact gives its mapping to each target language. Every thread calls it
and stores the result at an array index no other thread writes. The resulting
array of counts is what we compare: the interpreter produces one count per
thread, and a target system that groups threads differently produces different
counts. Such a mismatch is a difference in participation, the second condition
of WarpDRF, and thus would violate WarpDRF under the instantiated configuration.

\subsection{Configurations and Test Generation}
\label{sec:configurations-suites}

We instantiate three configurations. Two are derived in
Section~\ref{sec:instantiating}. The third, FMR, is stronger than any language
guarantees, and we include it to observe whether modern deployments fail it.
Each configuration names the dynamic blocks in which a primitive has a
guaranteed participant set.
\begin{itemize}
  \item \textbf{FMR} (full maximal reconvergence). Every dynamic block, so a
    primitive inside a divergent region runs with exactly the threads that
    branched into it.
  \item \textbf{SMR} (SPIR-V maximal reconvergence). Every dynamic block except
    at \texttt{switch}, where invocations that reach a common case through
    different selectors need not be grouped.
  \item \textbf{CUDA}. Only the dynamic blocks executed by the complete warp.
\end{itemize}
FMR is strictly stronger than SMR, which is strictly stronger than CUDA, so CUDA
represents the most portable semantics, but also the least flexible, allowing no
calls to non-uniform warp primitives.

The oracle interpreter implements FMR, the strongest configuration. However,
conformance does not test if participation is stronger than required, only if it
is not weaker. Thus, we do not need to implement ambiguous warp primitive
semantics. Instead, conformance testing can be instantiated over FMR by
constraining where warp primitives are allowed. For example, FMR allows warp
primitives anywhere in the program, as they would be unambiguous. SMR allows
them in all cases except for join nodes in switch statements. CUDA allows them
only in nodes with full warp participation. Furthermore, because our
configurations are totally ordered, the conformance tests are simply subsets of
each other.

A seeded generator produces a set of ``fuzzed'' conformance tests. It emits
conditionals, loops, switch statements, function calls, and early exits. Some
conditional expressions are intentionally statically obvious, so that the
compiler would attempt to flatten loops. Other conditionals cannot be known by
the compiler, since they depend on input arrays that are set to specific values.
Branches are set to be uniform or divergent at random. The generator at all
times tracks the dynamic execution graph and places warp primitives (along with
their logging functionality described above) at random places allowed by
configuration \(C\). The available parameters for the fuzzer are shown in
Table~\ref{tab:fuzzer-params}, along with the values we used in our study.

\begin{table}[t]
  \centering
  \scriptsize
  \setlength{\tabcolsep}{4pt}
  \renewcommand{\arraystretch}{1.05}
  \caption{Generator parameters used in the conformance study.}
  \label{tab:fuzzer-params}
  \begin{tabular}{@{}p{0.30\columnwidth}p{0.64\columnwidth}@{}}
    \toprule
    Parameter group & Setting \\
    \midrule
    Test campaign
      & 10,000 tests per configuration; seeds 0--9,999 \\
    Launch
      & One complete warp, $W\mathbin{\times}1\mathbin{\times}1$;
        $W\in\{16,32,64\}$ is queried from the target \\
    Control-flow bounds
      & Nesting depth $\leq3$; loop trips 1--4; switch cases 2--4 \\
    Helper bounds
      & Control nesting $\leq3$; call nesting $\leq2$;
        at least 3 control operations \\
    Generation probabilities
      & 0.3 each for break/continue, post-switch primitive,
        divergent helper call, switch-before-helper, and nested-helper loop \\
    Conditions
      & Thread-ID expressions or uniform-random input arrays \\
    \bottomrule
  \end{tabular}
\end{table}

\subsection{Conformance Across 16 Systems}
\label{sec:cross-platform}

Table~\ref{tab:device-fuzzing} shows the devices we test and the backends we
test them under, as several devices support multiple languages and thus
different compilers.

We start with conformance to CUDA, the weakest model. In our study, all GPUs and
their corresponding stacks honor this instantiation of WarpDRF, thus providing a
real portable programming interface. However, this comes at the cost of
functionality, as non-uniform warp primitives are not allowed at all.

Conformance implies three properties of warp execution that we could not
find stated in the CUDA documentation:
\begin{itemize}
  \item warps begin converged;
  \item there is no spurious divergence, so a path through the control-flow
    graph on which the warp is fully converged remains fully converged; and
  \item full reconvergence holds, so a warp that diverges is converged at
    the join where all of its threads meet.
\end{itemize}
Each is an exceptionally weak guarantee, but without them we cannot claim this
variant of WarpDRF is vendor supported in any way, even through informal prose.
It remains an empirically tested model.

Moving up to SMR, the two NVIDIA CUDA pairs fall out, which is not surprising.
The interesting point is that we expected a possible model between CUDA and SMR,
but in examining the failures we could not determine any reliable pattern that
CUDA devices support. Because our non-uniform warp primitives are implemented
using \texttt{\_\_activemask()} in CUDA, this means that
\texttt{\_\_activemask()} cannot be relied on to provide a deterministic set of
threads in any divergent code, an assumption that kernels might make. It is
interesting to note that Metal and HLSL also support the SMR model. Both define
a collective's result over the threads active when the operation executes, but
neither states a general rule for when threads from divergent paths must rejoin,
so their agreement with SMR is an empirical property of the tested
implementations.

It is typical for systems to implement stronger models than they document, and
ten of the sixteen provide FMR. The six that do not are every NVIDIA pair, the
PowerVR device, and the Adreno under GLSL. For four of them the only failures
are in switch statements; the two CUDA pairs already fell out at SMR.
Programmers targeting SPIR-V must therefore watch out for this construct: it is
documented, but it may still be surprising. Metal provides FMR on both devices,
and HLSL provides it everywhere except the NVIDIA laptop GPU. The stronger model
has the fewest surprises, but it also leaves an implementation the least
freedom, and investigating the performance implications of that would be
interesting future work. The mechanism is visible in the tests: when threads
with different selector values reach the same case through fallthrough, the
reference semantics groups them at that common case, while on the RTX 4070 GLSL
path the observed operation keeps them in separate groups.

\begin{table*}[t]
  \centering
  \footnotesize
  \setlength{\tabcolsep}{5pt}
  \renewcommand{\arraystretch}{0.92}
  \caption{Mismatches out of 10,000 tests for each configuration. A column
  places warp primitives only where that configuration guarantees a participant
  set, so a system conforms to a configuration when that cell and every cell to
  its left are green. Green cells contain no mismatches, while yellow cells
  contain one or more. The \texttt{-O0} row is an additional compiler
  configuration; the other rows cover 16 device/backend pairs.}
  \label{tab:device-fuzzing}
  \begin{tabular}{@{}llrrr@{}}
    \toprule
    Device & Backend & CUDA & SMR & FMR \\
    \midrule
    Intel Arc A770 & GLSL/Amber & \resultzero{0/10,000} & \resultzero{0/10,000} & \resultzero{0/10,000} \\
    Intel UHD 770 & GLSL/Amber & \resultzero{0/10,000} & \resultzero{0/10,000} & \resultzero{0/10,000} \\
    Mali G710 & GLSL/Amber & \resultzero{0/10,000} & \resultzero{0/10,000} & \resultzero{0/10,000} \\
    NVIDIA RTX 4070 & GLSL/Amber & \resultzero{0/10,000} & \resultzero{0/10,000} & \resultpartial{4,810/10,000} \\
    NVIDIA RTX 4070 & CUDA & \resultzero{0/10,000} & \resultpartial{747/10,000} & \resultpartial{1,431/10,000} \\
    NVIDIA RTX 4070 & CUDA, \texttt{-Xcicc=-O0} & \resultzero{0/10,000} & \resultzero{0/10,000} & \resultpartial{470/10,000} \\
    AMD RX 7900 XT & GLSL/Amber & \resultzero{0/10,000} & \resultzero{0/10,000} & \resultzero{0/10,000} \\
    AMD RX 7900 XT & HIP & \resultzero{0/10,000} & \resultzero{0/10,000} & \resultzero{0/10,000} \\
    Jetson T5000 (Thor) & CUDA & \resultzero{0/10,000} & \resultpartial{747/10,000} & \resultpartial{1,438/10,000} \\
    PowerVR DXT-48-1536 & GLSL/Amber & \resultzero{0/10,000} & \resultzero{0/10,000} & \resultpartial{24/10,000} \\
    Apple M4 & Metal & \resultzero{0/10,000} & \resultzero{0/10,000} & \resultzero{0/10,000} \\
    Apple M1 & Metal & \resultzero{0/10,000} & \resultzero{0/10,000} & \resultzero{0/10,000} \\
    Samsung Xclipse 940 & GLSL/Amber & \resultzero{0/10,000} & \resultzero{0/10,000} & \resultzero{0/10,000} \\
    NVIDIA RTX 4070 Laptop & HLSL/DirectX & \resultzero{0/10,000} & \resultzero{0/10,000} & \resultpartial{870/10,000} \\
    Intel Raptor Lake-S Mobile & HLSL/DirectX & \resultzero{0/10,000} & \resultzero{0/10,000} & \resultzero{0/10,000} \\
    Qualcomm Adreno X1-85 & GLSL/Amber & \resultzero{0/10,000} & \resultzero{0/10,000} & \resultpartial{1,370/10,000} \\
    Qualcomm Adreno X1-85 & HLSL/DirectX & \resultzero{0/10,000} & \resultzero{0/10,000} & \resultzero{0/10,000} \\
    \bottomrule
  \end{tabular}
\end{table*}

\subsection{WarpDRF: All or Nothing}
\label{sec:cuda-sensitivity}

The CUDA configuration appears strong, and one might be tempted to assume that
its components hold on their own: that a non-uniform primitive might not behave
well, but that it would not affect uniform primitives. This is
similar to assuming that benign data races will not impact data-race-free
accesses, which prior work has shown to be false~\cite{boehm2011benign}.

We show a similar situation here. A non-uniform primitive causes a fully uniform
primitive to have its participants split; the example is given in
Figure~\ref{fig:participant-question}. To make the problem even more nefarious,
we find that our observations change under different compiler configurations.

The kernel launches one warp of 32 threads and
runs nested loops whose trip counts depend on the thread: threads 0--15 run two
iterations of each, threads 16--31 run four. All 32 reach the primitive at the
four shared coordinates $(0,0)$, $(0,1)$, $(1,0)$ and $(1,1)$, where the
reference semantics gives 32 participants. On a Jetson AGX Thor the
\texttt{compute\_75} build returns 16 at all four, so 128 of the 320 counts the
kernel records disagree with the reference. The \texttt{compute\_110} build
returns 32 everywhere, as does a \texttt{compute\_75} variant that makes the
inner induction variable \texttt{volatile}.

The failure lands on full-warp coordinates. At every other coordinate only
threads 16--31 arrive, and those are interleaved with the shared ones, so the
kernel is outside the CUDA configuration and this is not a conformance failure.
It is worse: calling a primitive non-uniformly in the kernel changed
the result where the whole warp was present.

Thus WarpDRF should seek to get buy-in from implementors, as they are
responsible for maintaining this behavior across both hardware and compilers. We
believe our CUDA conditions are permissive, have conformance across a large
fuzzed data set, and could hopefully be incorporated into the specification.
Until then, our WarpDRF contract remains an empirical model, but we do not see
any more permissive semantics to provide.

\begin{figure}[t]
\centering
\coderule
\begin{lstlisting}[style=gpufigure,numbers=none,xleftmargin=0.5em,
                     aboveskip=3pt,belowskip=3pt]
// Launch: one 32-thread warp.
int bound = tid < 16 ? 2 : 4;
for (int outer = 0; outer < bound; ++outer)
  for (int inner = 0; inner < bound; ++inner)
    result[outer][inner][tid] = warp_count(true);
\end{lstlisting}
\begin{tikzpicture}[
  result/.style={
    draw,
    rectangle,
    align=center,
    text width=3.0cm,
    minimum height=17mm,
    inner sep=2pt,
    font=\scriptsize
  },
  flow/.style={
    -{Latex[length=1.4mm, width=1.1mm]},
    line width=0.4pt
  }
]
  \node[font=\scriptsize] (source) at (0,0) {same CUDA source};
  \node[result,fill=resultpartial] (sm75) at (-1.75,-1.35) {
    \texttt{compute\_75}\\
    6 PTX ballot sites\\
    8 SASS sync pairs\\
    \textbf{128 mismatches}};
  \node[result,fill=resultzero] (sm110) at (1.75,-1.35) {
    \texttt{compute\_110}\\
    1 PTX ballot site\\
    1 SASS sync pair\\
    \textbf{0 mismatches}};

  \draw[flow] (source.south west) -- (sm75.north);
  \draw[flow] (source.south east) -- (sm110.north);
\end{tikzpicture}
\coderule
\caption{A CUDA participant-count test launched with one 32-thread warp. Compiling the same
source for two virtual targets produces different PTX control-flow bodies and
different participant counts on Thor. Each thread writes a distinct location, so
the mismatches do not arise from conflicting ordinary memory accesses.}
\Description{A CUDA-like nested loop is compiled for compute 75 and compute
110. The compute 75 build has six PTX ballot sites, eight SASS synchronization
pairs, and 128 output mismatches. The compute 110 build has one PTX ballot site,
one SASS synchronization pair, and no mismatch.}
\label{fig:participant-question}
\end{figure}

Compiling the same source for \texttt{compute\_75} produces six PTX active-mask
and ballot sites and eight SASS \texttt{BSSY}--\texttt{BSYNC} pairs, whereas
\texttt{compute\_110} produces one of each. We executed each PTX body through
three compilation paths: direct Driver API loading, driver-link JIT compilation,
and standalone \texttt{ptxas}. Across 100 repetitions, every path preserves the
result of its input PTX: the \texttt{compute\_75} body produces 128 mismatches,
while the \texttt{compute\_110} body produces none. Changing only the PTX
\texttt{.target} declaration does not change either result. The resulting
machine code is not identical across all paths, so this experiment does not
identify one particular machine instruction as the cause. It does show that the
observed participant grouping follows the PTX control-flow structure selected
by the CUDA compiler, rather than the target declaration or the subsequent
compilation path.

\section{Checking Kernels Against the Contract}

Section~\ref{sec:realizing-warpdrf} found a stable full-warp boundary across the tested
implementations. We now ask whether deployed kernels follow that rule.
We produce the first static checker for an empirically validated warp
programming model by extending Faial for the portable CUDA configuration.
WarpDRF-Faial checks memory ordering and full-warp participation separately. We
apply it to the performance-critical CUDA kernels of llama.cpp, a widely used
open-source LLM inference engine. Most analyzed warp-using kernels satisfy
both conditions after source and caller validation; the remaining kernels
reduce to four source-level patterns.

\subsection{WarpDRF-Faial: a WarpDRF checker for CUDA}
\label{sec:faial}
\label{sec:extended-faial-protocols}

WarpDRF-Faial checks CUDA kernels for the two conditions of
Section~\ref{sec:abstract-executions}, reporting each condition separately
and reporting a kernel as WarpDRF only when both conditions succeed.
Our analysis builds on the GPU verification framework
Faial~\cite{faial:fmsd23,liew2024faialaa}.

\paragraph{Verifying Condition 1 (memory DRF)}
Our analysis assigns a \emph{block phase} and a \emph{warp phase} to each
memory access.
Two instructions alter the phases: a block-level barrier synchronization
(\texttt{\_\_syncthreads}) increments the block phase for every thread in the
block, and a warp barrier synchronization (\texttt{\_\_syncwarp}) increments
warp phase for the threads in that warp.
To check the DRF condition we apply the two-thread reduction
technique~\cite{pug:fse10}, where a conflict happens when the block phases
are equal and either the threads are in different warps, or the two accesses
are in the same warp and have equal warp phases.
\textbf{Alarm:} the counterexample (a program state) of the two accesses that
may race.

\paragraph{Verifying Condition 2 (unambiguous participation)}
The analysis checks that:
1) each warp primitive has a full-warp mask,
2) each branch guarding a warp primitive has a warp-uniform condition, and
3) each loop containing a warp primitive has warp-uniform bounds.
In short, the code that controls warp primitives must be warp-uniform, while
the remaining code can be warp-divergent.
The thread-dependent load in \texttt{block\_reduce} does not
control its subsequent warp reduction, so the checker accepts it.
\textbf{Alarm:} the warp primitive that may execute with only part of its
warp.

\subsection{WarpDRF-Faial on llama.cpp}
\label{sec:llama-eval}

\paragraph{Data selection}
At commit \texttt{42fc243}, llama.cpp contains 479 CUDA kernels across 64 source
files.
In this experiment we only consider the 95 kernels (out of 479) that use at
least one warp primitive.
In this experiment, we run WarpDRF-Faial once per kernel with a timeout of 60
seconds.
WarpDRF-Faial may report several alarms for one kernel.
We manually validate each alarm reported.
To validate the alarms, we group them by source location and reported cause,
inspect a representative specialization, and inspect caller constraints when
they can exclude the reported path.

\begin{table}[t]
  \centering
  \footnotesize
  \setlength{\tabcolsep}{5pt}
  \caption{Kernels checked per condition, with Condition~1 (memory) in rows
  and Condition~2 (participation) in columns. Pass means no alarms were
  found; Fail means one or more.}
  \label{tab:llama-conditions}
  \begin{tabular}{@{}lrrr@{}}
    \toprule
    \diagbox[width=13em,height=2em]{Memory}{Participation} & Pass & Fail & Total \\
    \midrule
    Pass & \textbf{65 (68.4\%)} & 0 (0.0\%) & 65 (68.4\%) \\
    Fail    & 29 (30.5\%) & 1 (1.1\%) & 30 (31.6\%) \\
    \midrule
    Total    & 94 (98.9\%) & 1 (1.1\%) & 95 (100\%) \\
    \bottomrule
  \end{tabular}
\end{table}

\paragraph{Findings}
Table~\ref{tab:llama-conditions} reports which kernels pass Conditions~1 and~2.
We found that 65 of the 95 kernels are automatically deemed WarpDRF.
We manually validated the remaining 30 kernels and found \textbf{18 buggy
kernels}; the other 12 have only false-positive or caller-discharged alarms.
This means that 77 kernels (65 + 12, or 81.1\%) satisfy both conditions.
The 18 buggy kernels group as MoE selection (12), MMVF partial accumulation
(4), ID compaction (1), and shared reduction-buffer reuse (1).

\paragraph{Alarm validation}
Across all kernels, WarpDRF-Faial found 67 memory alarms and 3 participation
alarms, summarized in Table~\ref{tab:alarm-families}.
We classify an alarm as \emph{confirmed} when the defect can be reproduced,
\emph{input-conditional} when the alarm requires values such as NaNs or
duplicate identifiers, \emph{caller-discharged} when current callers exclude
it, and \emph{false positive} otherwise.
The artifact provides the mapping from alarms to
locations and cause categories.

\begin{table}[t]
  \centering
  \scriptsize
  \setlength{\tabcolsep}{3pt}
  \caption{llama.cpp alarms grouped by source-level cause.}
  \label{tab:alarm-families}
\begin{tabular}{@{}p{0.36\columnwidth}p{0.17\columnwidth}rp{0.24\columnwidth}@{}}
  \toprule
  Cause category & Condition & Alarms & Outcome \\
  \midrule
  Vector-attention indexing & Memory & 10 & False positive \\
  Tiled-attention ownership & Memory & 2 & False positive \\
  ID output indexing & Memory & 3 & False positive \\
  Vector-attention participation & Participation & 3 & False positive \\
  \midrule
  MXFP4 scatter ownership & Memory & 2 & Caller-discharged \\
  Q8 scatter ownership & Memory & 27 & Caller-discharged \\
  SSM output layout & Memory & 2 & Caller-discharged \\
  MoE post-bias NaN & Memory & 12 & Input-conditional \\
  Duplicate-ID collision & Memory & 1 & Input-conditional \\
  MMVF shared-slot writes & Memory & 6 & Confirmed \\
  Shared reduction-buffer reuse & Memory & 1 & Confirmed \\
  ID shared-memory handoff & Memory & 1 & Confirmed \\
  \midrule \textbf{Total} & & \textbf{70} & \\
  \bottomrule
\end{tabular}

\end{table}

\emph{False positives.}
The 18 false-positive alarms have four causes. Ten vector-attention memory
alarms and two tiled-attention memory alarms arise because the checker cannot
prove that the kernels' template-dependent index calculations assign distinct
buffer slices or output elements to different threads. Inspecting those index
calculations shows that the reported accesses cannot refer to the same
location. Three ID-compaction memory alarms arise because the checker treats a
full-warp reduction result as potentially different in each thread. In the
kernel, every thread receives the same reduction result and then uses its thread
identifier to access a different output element. The remaining three alarms
concern participation: the checker cannot prove that block-uniform bounds and
warp-aligned thread ranges select complete warps. The supported template and
launch configurations make those bounds multiples of the warp width, so each
affected warp executes the primitive with either all of its threads or none of
them.

\emph{Caller-discharged alarms.}
The 31 caller-discharged alarms arise because the checker analyzes a kernel
under a weaker input contract than its current callers provide. For the MXFP4
and Q8 scatter kernels, it therefore allows two thread groups to map to the
same compact output row and report conflicting writes. The callers construct
an inverse map in which each compact row has a unique owner; within that row,
thread guards assign each packed field to one thread. These properties make the
reported destinations distinct. For the SSM kernel, the checker similarly
allows the result and state pointers to overlap. The caller's tensor layout
places the state region after the result region, so the reported accesses lie
in separate regions. These conclusions apply to the current callers rather
than to arbitrary calls satisfying only the kernel's local parameter types.

\emph{Input-conditional alarms.}
A post-bias NaN can make MoE-selection threads choose different
experts and write the same output. Duplicate IDs can make two threads
write different values to one shared entry in the ID-compaction kernel.

\emph{Confirmed alarms.}
In MMVF, every thread of a warp writes one shared slot without ordering. The
shared reduction-buffer alarm arises when a later reduction overwrites the
buffer while another warp is still reading it. The ID shared-memory handoff is
unordered as its intervening shuffle coordinates register values but
does not order shared memory.
We reported this bug in
\href{https://github.com/ggml-org/llama.cpp/issues/28376}{issue~\#28376}
and proposed a warp memory barrier before the shared-array reads.
Merged upstream \href{https://github.com/ggml-org/llama.cpp/pull/28475}{PR~\#28475}
subsequently fixed the handoff by adding that barrier.

\section{Related Work}

\paragraph{Scoped memory models and data-race freedom}
Weak ordering established the DRF style of reasoning: programmers use
synchronization to order conflicting accesses and then reason in a simpler
reference model~\cite{adve1990weak}. Heterogeneous race-free models adapt that
argument to hierarchical machines by attaching scopes to synchronization and
memory operations~\cite{hower2014hrf,sinclair2017rats}, and formal work has
characterized the guarantees that heterogeneous systems, PTX, and Vulkan
actually
provide~\cite{alglavegpuconcurrency,dan-ptx-mem-model,haining2025gpuconsis}.
All of this begins with synchronization operations whose participants are
already fixed. WarpDRF keeps the DRF argument but not that premise: a warp
primitive orders only the threads that execute it together, so its participants
must be established before its memory ordering can be used.

\paragraph{GPU verification and race detection}
Earlier work compares interleaving and lockstep GPU semantics and establishes
their agreement for race-free
programs~\cite{collingbourne2013interleaving,kojima2017hoare}. GPUVerify was
extended to reason about implicit, coarse-grained synchronization within fixed
warps and about atomics~\cite{bardsley2014warps}. Dynamic tools including
BARRACUDA, CURD, and HiRace detect races in GPU executions, while Simulee
targets CUDA synchronization bugs through memory-access
modeling~\cite{eizenberg2017barracuda,peng2018curd,jacobson2024hirace,
wu2020simulee}. These approaches do not define the participant group of an
explicit warp primitive under divergent control flow; WarpDRF makes that group
part of the contract.

\paragraph{Testing implementations against specifications}
A specification is only useful if implementations follow it, and a line of work
establishes that empirically. Stress testing has exposed weak and undocumented
memory behavior in deployed GPU applications and
devices~\cite{tyler2016weakmem,jake2020memcons}, mutating litmus tests measures
how well a suite discriminates between candidate
models~\cite{reese-memorycons}, and large campaigns survey conformance across
many devices at once~\cite{reesegpuhabor}. Related work has found shader
compiler defects with generated programs~\cite{ally-fuzz} and characterized
forward-progress guarantees that no specification
states~\cite{tyler-forwardprogress}. These campaigns target memory consistency
and progress, where a reference model already fixes the expected result.
Section~\ref{sec:realizing-warpdrf} applies the method to participation, which
first required a contract to supply one.

\paragraph{Warp execution and participation}
Warp-level cooperation is a standard performance idiom, with libraries
providing warp-wide reductions, scans, and data exchange~\cite{cub} and
irregular kernels built directly on warp behavior~\cite{merrill2012graph}. What
those operations group is less settled. NVIDIA's independent thread scheduling
weakens the assumption that a warp executes in
lockstep~\cite{nvidia_volta_tuning_2025}, and hardware and compiler work treats
grouping as an optimization variable, regrouping threads after divergence or
delaying execution to increase
convergence~\cite{fung-dynamic-warp,ElTantawy2018warp,oliver2020}. This line
studies grouping as a performance question.

A separate line makes participation precise. SIMT-Step determines the threads
that execute a subgroup operation together from dynamic basic
blocks~\cite{chen2026simtstep}, SPIR-V control-flow work and maximal
reconvergence expose the structured joins such reasoning
needs~\cite{vaileios-controlflow,spirv:max-reconvergence}, and uniformity and
barrier-divergence analyses identify where a collective may execute
safely~\cite{kamil2006concurrency,betts2012gpuverify,leejones2026uniformity}.
Verification tools cover fixed-warp synchronization,
producer-consumer patterns, and source-level correctness
requirements~\cite{bardsley2014warps,sharma2015weft,martinez2026kuiper}. These
describe how participants form. WarpDRF makes participation a condition of a
data-race-freedom guarantee, then checks whether implementations and deployed
kernels satisfy it.

\section{Conclusion}

This paper introduced WarpDRF, a data-race-freedom contract for GPU warp
programming. WarpDRF combines scoped memory ordering with predictable
participation in each warp primitive. It
provides a portable model for full-warp programs and a separate model for
languages that define partial-warp reconvergence. We implemented WarpDRF
through an MLIR interpreter, a cross-platform conformance suite, and a
warp-aware extension to Faial. Across the tested systems, full-warp programs
behaved consistently, whereas partial participation varied across languages,
compilers, and targets. Applying WarpDRF-Faial to llama.cpp exposed
complementary memory and participation problems. Together, these results
provide a rigorous foundation for specifying, reasoning about, and checking
warp programs across GPU languages and implementations.

\begin{acks}
We disclose the use of generative AI in this work. AI aided in tool
development, evaluation, and data analysis. Humans rigorously evaluated the
integrity of all artifacts reported here. Humans organized the presentation of
the work, and agents helped with drafting and polishing. Text polishing used
the Copilot IDE set to Claude Opus 5, and the authors used other coding
assistants running GPT-5.6 and Claude Opus 5.
\end{acks}

\bibliographystyle{ACM-Reference-Format}
\bibliography{software}

\end{document}